\documentclass{amsart}
\usepackage{amssymb}

\usepackage{hyperref}

\usepackage[usenames,dvipsnames]{color}
\newcommand\ownremark[1]{}

   {\end{trivlist}\platz}

\newtheorem{theorem}{Theorem}
\newtheorem{proposition}[theorem]{Proposition}
\newtheorem{lemma}[theorem]{Lemma}

\newtheorem{conjecture}[theorem]{Conjecture}

\renewcommand {\epsilon}{\varepsilon}
\newcommand{\CinftydOmega}{C^\infty(\partial\Omega)}
\newcommand{\eps}{\varepsilon}
\newcommand{\Rtilde}{\tilde{R}}
\newcommand{\bfg}{{\bf g}}
\newcommand{\geucl}{\bfg_{\text{eucl}}}
\newcommand{\Ocomp}{\R^n\setminus\overline{\Omega}}
\newcommand{\dvol}{\text{dvol}}
\newcommand{\epsdot}{\dot{\eps}}
\newcommand{\epsddot}{\ddot{\eps}}
\newcommand{\dn}{\partial_n}

\newcommand{\Cinfty}{C^\infty}
\newcommand{\scpr}[3]{\langle #1 \, |\, #2 \, | \, #3 \, \rangle } 
\newcommand{\calNdot}{\dot{\calN}}
\newcommand{\gdot}{\dot{g}}
\newcommand{\Hdot}{\dot{H}}

\newcommand\calH{{\mathcal H}}

\newcommand\calN{{\mathcal N}}
\newcommand\calO{{\mathcal O}}

\newcommand\calR{{\mathcal R}}

\newcommand\Id{{\rm Id}}

\newcommand\loc{{\rm loc}}

\newcommand{\Span}{\operatorname{span}}

\newcommand\vol{{\operatorname{vol}}}

\newcommand {\R}{\mathbb{R}}

\newcommand {\N}{\mathbb{N}}

\newcommand{\bfn}{{\boldsymbol{n}}}
\newcommand{\bfq}{{\boldsymbol{q}}}
\newcommand{\epsilondot}{\dot{\epsilon}}
\newcommand{\udot}{\dot{u}}
\newcommand{\uddot}{\ddot{u}}
\newcommand{\dO}{{\partial\Omega}}

\newcommand{\bfndot}{{\boldsymbol{\dot{n}}}}
\newcommand{\bfnddot}{{\boldsymbol{\ddot{n}}}}

\newcommand{\nablad}{\nabla_\partial}
\newcommand{\Deltad}{\Delta_\partial}
\renewcommand{\div}{\operatorname{div}}
\newcommand{\divd}{\div_\partial}
\newcommand{\intdO}{\int_\dO}

\newcommand{\Omegabar}{\overline{\Omega}}
\newcommand{\uhat}{\hat{u}}
\newcommand{\ghat}{\hat{g}}
\newcommand{\Rinfty}{\R\cup\{\infty\}}
\newcommand{\calHbar}{\overline{\calH}}
\newcommand{\Gtilde}{\tilde{G}}

\begin{document}
\title{The plasmonic eigenvalue problem}
\author{
Daniel Grieser}
\address{Institut f\"ur Mathematik, Carl-von-Ossietzky Universit\"at Oldenburg, 26111 Oldenburg, Germany}
\email{daniel.grieser@uni-oldenburg.de}
\date{\today}

\subjclass[2010]{78A99 
			35Pxx 
			31A05 
     }

\keywords{Dirichlet-Neumann operators, pseudodifferential operators, perturbation theory, eigenvalue problems}

\begin{abstract}
A plasmon of a  bounded domain $\Omega\subset\R^n$ is a non-trivial bounded harmonic function on $\R^n\setminus\partial\Omega$ which is continuous at $\partial\Omega$ and whose exterior and interior normal derivatives at $\partial\Omega$ have a constant ratio. We call this ratio a plasmonic eigenvalue of $\Omega$. Plasmons arise in the description of electromagnetic waves hitting a metallic particle $\Omega$. We investigate these eigenvalues and prove that they form a sequence of numbers converging to one. Also, we prove regularity of plasmons, derive a variational characterization, and prove a second order perturbation formula. The problem can be reformulated in terms  of Dirichlet-Neumann operators, and as a side result we derive a formula for the shape derivative of these operators.
\end{abstract}

\maketitle
\tableofcontents

\section{Problem and results}
Let $\Omega\subset\R^n$, $n\geq 3$
be a bounded domain with smooth boundary $\partial\Omega$ and connected complement $\R^n\setminus\Omega$.
For functions $u_-$ on $\Omegabar$ and $u_+$ on  $\R^n\setminus\Omega$, assumed to be smooth up to the common boundary $\partial\Omega$, we consider the following problem:
\begin{eqnarray}
\Delta u_\pm=0\quad &&\text{in}\quad \Omega \text{ resp. }\R^n\setminus \Omegabar\label{I1}\\
u_--u_+=0 \quad &&\text{on}\quad  \partial\Omega\label{I2}\\
\varepsilon\partial_n u_- +\partial_n u_+ =0\quad &&\text{on}\quad  \partial\Omega\label{I3},\\
u_+(x) = O(|x|^{2-n}) \quad &&\text{as}\quad |x|\to\infty \label{I4}
\end{eqnarray}
where $\partial_n u_\pm=\boldsymbol{n}\cdot \nabla u_\pm$ is the derivative of $u_\pm$ in the direction of the unit  normal $\boldsymbol{n}$ pointing out of $\Omega$. Here $\eps\in\R$ is a constant. Thus, the function $u$ on $\R^n$ defined by $u_-$ on $\Omega$ and by $u_+$ on $\R^n\setminus\Omega$ is continuous at $\dO$, but its normal derivative may have a jump, and the ratio of inner and outer normal derivative is required to be constant on $\dO$. It will be convenient to also allow $\eps=\infty$, in which case \eqref{I3} is interpreted as $\partial_n u_- \equiv 0$, $\partial_n u_+\not\equiv 0$.

We call $\varepsilon\in\Rinfty$ a \emph{plasmonic eigenvalue} of $\Omega$ if the system \eqref{I1}-\eqref{I4} has a solution $u\not\equiv 0$, and then $u$ is called a plasmonic eigenfunction or (surface) {\em plasmon } for $\varepsilon$. This name is justified by the origin in the physics of plasmons, and by the fact, established below, that $\eps^{-1}$ is indeed an eigenvalue of a certain operator on $\partial\Omega$. For a given plasmonic eigenvalue we call the linear space of solutions $u$ its plasmonic eigenspace.
What is usually denoted $\varepsilon$ in the physics literature is $-\varepsilon$ in our notation. Our choice of sign is motivated by the fact that $\varepsilon>0$ then. 
A corresponding problem may also be formulated for $n=2$, but this requires a few minor adjustments, see Section \ref{subsec:2d}.

Plasmons have been studied extensively in the physics literature, though principally for simple geometries, where explicit calculations are possible, and for more general domains by numerical methods. See the end of the introduction for more on the physics background and references.
The purpose of this paper is to put the problem in its natural mathematical setting and to bring some more refined mathematical techniques to bear on the study of plasmonic eigenvalues. This will show that also for general domains interesting conclusions can be drawn even without numerical calculation. In particular, we will employ microlocal analysis to analyze the asymptotic behavior of the plasmonic eigenvalues, and differential geometry to calculate the perturbation of plasmonic eigenvalues under perturbation of the domain, to second order. This is the order needed for investigating the behavior under random perturbations of the domain, a topic of physical interest. Also, the perturbation results can be useful in obtaining good approximations for non-infinitesimal deformations of the sphere, for example, see  \cite{BieEtAl09}.

\medskip

We now state our results. First, we have basic results on existence and regularity of plasmons, and asymptotics of plasmonic eigenvalues.
Denote 
 $$\calH = \{u \text{ satisfying } \eqref{I1}, \eqref{I2}, \eqref{I4}\},$$
and on $\calH$
consider the scalar product (for simplicity we assume all functions to be real-valued)
\begin{equation}
\label{eqn:scprod+} 
 (u,v)_+ := \int_{\Omega^c} \nabla u\cdot\nabla v 
\end{equation}
where $\Omega^c=\R^n\setminus\Omega$. This is finite since harmonic functions on $\Omega^{c}$ satisfying \eqref{I4} have $\nabla u (x)= O(|x|^{1-n})$. Nondegeneracy follows from the unique solvability of the Dirichlet problem for $\Omega$. Thus $\calH$ with $(\cdot,\cdot)_+$ is a pre-Hilbert space.

\begin{theorem}[Existence, completeness, asymptotics] \label{thm:existence+asymp}
Let $\Omega\subset\R^n$, $n\geq3$ be a bounded domain with smooth boundary. Then the finite plasmonic eigenvalues of $\Omega$ form a sequence $\eps_1,\eps_2,\dots$ of positive numbers. For $\eps\neq 1$  the corresponding eigenspaces are finite dimensional.

Furthermore, $\eps_k\to 1$ as $k\to \infty$. Also, if
$E_\eps$ is the space of plasmons $u$ with eigenvalue $\eps$ then we have an orthogonal decomposition
$$ \calH = \bigoplus_{\eps\in\Rinfty} E_\eps $$
Also, $\eps=\infty$ is a plasmonic eigenvalue, with one-dimensional eigenspace
$$ E_\infty = \{u\in\calH:\, u_-\text{ is constant}\} $$
\end{theorem}
See Section \ref{subsec:2d} for the case $n=2$.

One may also formulate the eigenvalue problem with weaker regularity assumptions. We then have the following result.
\begin{theorem}[Regularity] \label{thm:regularity}
Let $u$ be a solution of $\eqref{I1}-\eqref{I4}$ with
 $u_-\in H^s(\Omega)$, $u_+\in H^{s}_{\loc}(\Omega^c)$ for some $s>3/2$.

If $\eps\neq 1$ then $u$ is $C^\infty$ on $\R^n\setminus\dO$ up to the boundary from each side.
\end{theorem}
Here $H^s$ is the Sobolev space of order $s$, and $H^s_\loc(\Omega^c)$ is the space of functions whose restriction to $K_R=\Omega^c\cap\{|x|<R\}$ is in $H^s(K_R)$ for all $R$.
The range of $s$ is chosen so that the restrictions of $u$ and $\dn u$ to the boundary are well-defined by the Sobolev restriction theorem. 

Examples show that the statement of the theorem is not true in general for $\eps=1$, see Section \ref{subsec:disk}. 
Note that the only issue is boundary smoothness since harmonic functions are always smooth in the interior of the domain.

The main ingredient in the proof of Theorems \ref{thm:existence+asymp} and \ref{thm:regularity} is a reformulation as a problem on $\dO$ in terms of Dirichlet-to-Neumann operators, see Section \ref{subsec:dir-neu}. This is closely related to the fact that for a harmonic function $u_-$ on $\Omega$ we have
\begin{equation}
\label{eq:ibp}
\int_{\dO} u_- \partial_n u_- = \int_\Omega |\nabla u_-|^2 =: \|u\|_-^2
\end{equation}
which follows from Green's formula. Similarly, for a harmonic function $u_+$ on $\R^n\setminus\overline{\Omega}$ satisfying \eqref{I4}, we have
\begin{equation}
\label{eq:ibp+}
-\int_{\dO} u_+ \partial_n u_+ =  \int_{\Omega^c} |\nabla u_+|^2
=: \|u\|_+^2
\end{equation}
Theorem \ref{thm:existence+asymp} was formulated in terms of the scalar product $(\cdot,\cdot)_+$ corresponding to the norm $\|\cdot\|_+$ on $\calH$. The orthogonality there also holds with respect to the bilinear form $(u,v)_- := \int_\Omega \nabla u \cdot \nabla v$ corresponding to $\|\cdot\|_-$ which, however, is not positive definite -- it is zero on $E_\infty$.

Equations \eqref{eq:ibp} and \eqref{eq:ibp+} imply that for a plasmon $u$ the plasmonic eigenvalue equals
\begin{equation}
\label{eq:el-functional}
\eps = R_p (u),\qquad R_p (u) := \frac{\int_{\Omega^c} |\nabla u_+|^2}{\int_{\Omega} |\nabla u_-|^2}
\end{equation}
Denote by $\Hdot^1(\R^n)=\{u\in L^2_\loc(\R^n): \nabla u \in L^2(\R^n)\}$ the first homogeneous Sobolev space. Then we have:

\begin{theorem}[Variational principle] \label{thm:euler-lagrange}
\eqref{I1}-\eqref{I3} is the Euler-Lagrange equation of the functional $R_p$
on  $\Hdot^1(\R^n)\setminus 0$. More precisely, the finite plasmonic eigenvalues are precisely the positive critical values of $R_p$ and the associated plasmons are the corresponding critical points of $R_p$.
\end{theorem}
Note that zero is also a critical value of $R_p$, with any smooth function with non-empty support contained in  $\Omega$ as critical point. Of course, these do not correspond to plasmons.

In this theorem we use $H^1_\loc$ instead of $H^1$ since functions satisfying \eqref{I4} need not be in $L^2(\R^n)$ for $n\leq 4$.
\medskip

Next we consider the behavior of plasmonic eigenvalues under perturbation of the domain. Thus, fix $\Omega$ and let $a:\partial\Omega\to\R$ be a smooth function, and for real numbers $h$ close to zero consider the bounded domain $\Omega(h)$ with smooth boundary
\begin{equation} \label{eq:def Omega_h}
 \partial\Omega(h) = \{x + h a(x)\,\bfn(x):\, x\in \partial\Omega\}.
\end{equation}
Thus, $\Omega(0)=\Omega$ and $\Omega(h)$ is obtained by shifting the boundary in the normal direction by $ha$. $h$ is the order of magnitude of the shift and $a$ is a 'shape' function.

In the following theorem, we state the perturbation formulas only in the physically interesting case $n=3$. Generalization to any $n\geq2$ is straightforward.

{\bf Notation from differential geometry:} $H$ and $K$ denote the mean and Gau{\ss} curvature of $\dO$, $W_0=W-HI$ the trace free part of the Weingarten map $W$, and $\div_\partial$ and $\nabla_\partial$ the divergence and gradient operator in the surface $\dO$. The definitions of these quantities are recalled in Section \ref{secperturb}. Also,
\begin{equation}
\label{eq:def scalar product dO}
 \langle g,g' \rangle := \int_{\dO} g g' \,dS,\qquad g,g'\in \Cinfty(\dO)
\end{equation}
where $dS$ is surface measure (all functions are assumed real-valued), and for an operator $A$ on $\Cinfty(\dO)$ we let
$$ \scpr{g}{A}{g'} := \langle g, Ag' \rangle $$
Similarly, for vector fields $V,V'$ on $\dO$ and operator $A$ acting on vector fields, $\scpr V  A {V'} := \langle V, AV' \rangle :=\int_{\dO} V\cdot AV'\,dS$. All operators appearing below are self-adjoint.

\begin{theorem} \label{thm:perturbation}
Let $\eps\neq 1$ be a finite plasmonic eigenvalue of $\Omega\subset\R^3$ with eigenspace $E$. Then there are $h_0>0$ and real analytic functions
$h\mapsto \eps^{(i)}(h)$, $h\mapsto u^{(i)}(h)$ defined for $|h|<h_0$, $i=1,\dots,\dim E$, such that for each $h$ the numbers $\eps^{(i)}(h)$, $i=1,\dots,\dim E$ are plasmonic eigenvalues of $\Omega(h)$  with eigenfunctions $u^{(i)}(h)$, and $u^{(1)}(0),\dots,u^{(\dim E)}(0)$ are a basis of $E$.

For a fixed analytic branch $\eps(h)$, $u(h)$ of eigenvalue and eigenfunction with $\eps(0)=\eps$ and $\|u(h)\|_-=1$ for each $h$ we have for $\dot{\varepsilon}:=\frac{d}{dh}{\varepsilon}(0), \ \dot{u}(x):=\frac{d}{dh}{u}(x,0),\
\ddot{\varepsilon}:=\left(\frac{d}{dh}\right)^2{\varepsilon}(0)
$ the formulas:
\renewcommand{\arraystretch}{1.5}
\begin{align}
\epsdot & = q_1(u) := (\eps+1) \big[
- \scpr{\nablad u_-}{a}{\nablad u_-} +
\eps \scpr {\dn u_-}{a}{\dn u_-}\big] \label{eq:epsdot}\\
&= (\eps  + 1) \int_{\dO} a \left[-|\nablad u_-|^2 + \eps (\dn u_-)^2 \right]\, dS \notag \\
\label{eq:epsddot}
\tfrac12\epsddot & =
\begin{array}[t]{c@{\,\langle\,}l@{\,|\,}c@{\,|\,}l@{\,\rangle}}
& \nablad u_-  & {-\epsdot a - (\eps+1)a^2 W_0}  &
{\nablad  u_-} \\ 
+ &  \nablad u_- & (\eps^2-1) a & {\nablad(a\dn u_-)} \\
+ & {\nablad u_-} & {-(\eps+1)a} & {\nablad\udot_-}\\
+ & u_- & -\epsdot & \dn \udot_-\\
+  & {\dn u_-} & {a\eps (\epsdot + (\eps+1)aH)} & {\dn u_-} \\
+ &  {\dn u_-} & {\eps(\eps+1)a} & {\dn \udot_-}
\end{array}
\end{align}
Here, $\udot$ is a solution of the inhomogeneous system
\begin{equation}
\label{eq:udot}
\begin{aligned}
\Delta \dot{u} &=0\quad &&\text{in}\quad \R^3\setminus
\partial\Omega,\\
\dot{u}_--\dot{u}_+ &= -(\varepsilon+1) a\partial_n u_-\quad &&\text{on}\quad \partial\Omega,\\
\varepsilon \partial_n \dot{u}_- + \partial_n
\dot{u}_+&=-\dot{\varepsilon}\partial_n u_-+(\varepsilon+1)
\divd(a\nablad u_-)\quad&&\text{on}\quad
\partial\Omega
\end{aligned}
\end{equation}

Also, $u(0)$ is a critical point of the functional $R_p'(u) = \frac{q_1(u)}{\|u\|_-^2}$ on $E$, with $q_1$ defined in \eqref{eq:epsdot}, with critical value $\epsdot$.
\end{theorem}
The meaning of analyticity of $h\mapsto u^{(i)}(h)$ will be explained in Section \ref{subsec:analytic}.
Note that $\udot$ in \eqref{eq:udot} is not unique, but the value of $\epsddot$ calculated from \eqref{eq:epsddot} is independent of the choice of $\udot$.
There is a similar inhomogeneous system determining $\uddot$, see \eqref{eqn1b}-\eqref{eqn3b}. In the proof of the theorem, the formulas for $\epsdot$, $\epsddot$ are obtained from the requirement that these systems have a solution.
The last statement of the theorem helps to identify which $u$ to use in \eqref{eq:epsdot} in case of degeneracy, i.e. $\dim E>1$.

As a side result we also obtain a perturbation formula for the Dirichlet-Neumann operator which may be of independent interest, see Theorem \ref{thm:DN-perturb}. This is sometimes called the shape derivative of the Dirichlet-Neumann operator. Recently there has been some interest in such results in the context of water waves, see \cite{lannes} for example.

\subsection*{Physics background and previous work}
The problem \eqref{I1}-\eqref{I4} with $n=3$ arises in the physics of surface plasmon polaritons. These are waves which couple electromagnetic fields to the
electron gas of a metal, and thus propagate along the interface of a metal
and a dielectric material or vacuum. 
As outlined in several recent reviews,
the particular interest in these excitations stems from the fact that by
altering the structure of the metal's surface, the properties of these
plasmons, including their interaction with light, can be specifically
tailored~\cite{Atwater07,BarnesEtAl03,Maier07,
PitarkeEtAl07,ZiaEtAl06}. This offers a high potential for developing new types of
photonic devices, combining the speed of photonics with the nanoscale
dimensions of electronics. Possible applications include, among others,
plasmonic chips, data storage, nanolithography, and highly sensitive
molecular detectors.

Our setting models a metallic particle whose shape is given by $\Omega$, surrounded by vacuum. The function $u$ is the potential of the electric field, and the equations arise from Maxwell's equations in the quasistatic approximation, which assumes constancy of the magnetic field, and is justified if the wavelength is much larger than the dimensions of the particle, which is the case for particles on the nano scale. 
See \cite[Chapter 7]{landau-lifschitz} for details. In the context of surface plasmons the quasistatic approximation has often been used in the literature, see \cite{KlimovGuzatov07b,KlimovGuzatov07a} for example, and its validity has been verified experimentally.

 The boundary conditions \eqref{I2}, \eqref{I3} are obtained from continuity of the tangential component of the electric field and of the normal component of the electric displacement, see \cite[Chapter 4.4]{griffiths}. The number $-\eps$ is the relative permittivity of the material of the particle, which  is coupled to the frequency by a law specific to the material. Thus, plasmonic eigenvalues correspond to frequencies of electromagnetic waves at which there is a strong excitation of the electron plasma in the metallic particle $\Omega$.
The name plasmonic eigenvalue was coined in  \cite{GrieserRuting09}, where also the asymptotics $\eps_k\to1$ was proved. The paper
 \cite{BieEtAl09} contains the first order perturbation formula \eqref{eq:epsdot}.
\medskip

The spectrum of the Laplacian on a bounded domain, and its relation to the geometry of the domain, has been studied intensively for many years, so it may be worth while to compare it to our problem. There are similarities, for example the existence of discrete spectrum  and the presence of the low energy and the high energy regime, which require very different mathematical techniques. There are also clear differences, for example the asymptotic behavior and the fact that plasmonic eigenvalues are scale invariant while the Laplacian eigenvalues are not.

\subsection*{Acknowledgement}
I am grateful to Gohar Harutyunyan for checking the perturbation formula \eqref{eq:epsddot} independently. 
\section{Basic observations}
\label{sec:basics}
\subsection{Relation to Dirichlet-Neumann operators; proofs of Theorems \ref{thm:existence+asymp} and \ref{thm:regularity}}
\label{subsec:dir-neu}
Theorems \ref{thm:existence+asymp} and \ref{thm:regularity} are proved by reformulating the plasmonic eigenvalue problem as an 'honest' eigenvalue problem for an operator on $C^\infty(\partial\Omega)$ which is built from Dirichlet-to-Neumann operators.
For basics on harmonic functions in $\Omega$ and $\R^n\setminus\Omegabar$, and for proofs of some of the facts stated below, see for example \cite{folland}. The behavior at infinity of harmonic functions may, for example, be analyzed by using a Fourier decomposition on large circles.

The (interior) Dirichlet-to-Neumann operator of $\Omega$ maps Dirichlet to Neumann data of a harmonic function. That is,
$$ \calN_- : \CinftydOmega \to \CinftydOmega,\quad g \mapsto \partial_n u_-$$
where $u_-$ is the unique solution of the Dirichlet problem $\Delta u_-=0$ in $\Omega$, $u_{-|\partial\Omega} = g$. Similarly, the exterior Dirichlet-to-Neumann operator is
$$ \calN_+ : \CinftydOmega \to \CinftydOmega,\quad g \mapsto \partial_n u_+$$
where $u_+$ is the unique solution of the Dirichlet problem $\Delta u_+=0$ in $\R^n\setminus\Omegabar$, $u_{+|\partial\Omega} = g$, $u_+(x)= O(|x|^{2-n})$ as $|x|\to\infty$. It is standard that this decay requirement already follows from boundedness of $u_+$.

$\calN_{\pm}$ are self-adjoint operators on $L^2(\partial\Omega)$ with the standard scalar product \eqref{eq:def scalar product dO}. The symmetry of $\calN_-$ follows immediately from Green's formula, and for $\calN_+$ can be reduced to that of $\calN_-$ by use of the Kelvin transform.

Note that while $\calN_-$ is a positive operator by \eqref{eq:ibp}, $\calN_+$ is negative because the normal $\bfn$ points into $\Ocomp$ rather than out.
\begin{proposition}\label{prop:dirichlet-neumann}
The plasmonic eigenvalue problem \eqref{I1}-\eqref{I4} is equivalent to the generalized eigenvalue problem
\begin{equation}
\label{eq:gen ev pr}
(\eps \calN_- + \calN_+) g = 0,
\end{equation}
in the sense that smooth solutions $(\eps,u)$ of \eqref{I1}-\eqref{I4} are in 1-1 correspondence with smooth solutions $(\eps,g)$ of \eqref{eq:gen ev pr}, via
$g = u_{|\partial\Omega}$.
\end{proposition}
\begin{proof}
By the previous discussion we have a 1-1 correspondance
\begin{equation}
\label{eq:1-1 corr}
\begin{aligned}
\calH & \longleftrightarrow
C^\infty(\dO) \\
u &\mapsto u_{|\dO} \notag
\end{aligned} 
\end{equation}
Under this correspondence \eqref{I3} is equivalent to \eqref{eq:gen ev pr}.
\end{proof}

For $g,g' \in C^\infty(\dO)$ we denote 
$$ 
 (g,g')_+ := - \langle g, \calN_+ g' \rangle
$$
If $u,u'\in\calH$ are the functions on $\R^n$ corresponding to $g,g'$ under the correspondence \eqref{eq:1-1 corr} then Green's formula, cf. \eqref{eq:ibp+}, gives 
\begin{equation}\label{eqn:isometry}
 (g,g')_+  = (u,u')_+  
\end{equation}
where the latter is defined in \eqref{eqn:scprod+}. This justifies the notation and shows that $(\cdot,\cdot)_+$ is a scalar product on $C^\infty(\dO)$. The latter fact also follows from the well-known fact that $\calN_+$ is invertible.

The problem \eqref{eq:gen ev pr} can easily be transformed into a standard eigenvalue problem: The plasmonic eigenvalues are the reciprocals of the eigenvalues of the operator
$$A=-\calN_+^{-1}\calN_-.$$
This operator is symmetric with respect to $(\cdot,\cdot)_+$. Introduce the Hilbert space
$$ \calHbar := \text{ completion of $C^\infty(\partial\Omega)$ with respect to $(\cdot,\cdot)_+$.}  $$

\begin{proof}
[Proof of Theorem \ref{thm:existence+asymp}]
Positivity of plasmonic eigenvalues follows from \eqref{eq:el-functional}: $R_p(u)$ must be non-zero since otherwise $u_+\equiv 0$, so $u_{-|\dO}=0$, hence $u\equiv 0$.

We use some concepts and facts from microlocal analysis. See \cite{folland} for background. It is well-known that $\calN_-,\calN_+$ are first order classical pseudodifferential operators on $C^\infty(\partial\Omega)$ with principal symbols
$$ \sigma(\calN_-) = |\xi|,\quad  \sigma(\calN_+) = -|\xi|. $$
The fact that $\calN_+$ is an invertible elliptic first order operator implies that $\calHbar=H^{1/2}(\partial\Omega)$, the Sobolev space of order $1/2$. By the standard theory of pseudodifferential operators,
$A=-\calN_+^{-1}\calN_-$ is a zeroth order pseudodifferential operator with principal symbol $-(-|\xi|^{-1}) \cdot |\xi|=1$, hence $A = I + R$ where $R$ is a pseudodifferential operator of order $-1$. Therefore, $R$ is a compact operator on $\calHbar$, and since $A$ is selfadjoint, so is $R$, so $R$ has real nonzero eigenvalues $r_1,r_2,\dots \to 0$ of finite multiplicity, and may have a zero eigenvalue $r_0=0$ of any multiplicity. The eigenspaces span $\calHbar$ and are orthogonal with respect to $(\cdot,\cdot)_+$. Then $a_k=1+r_k$ are the eigenvalues of $A$, so $a_k^{-1}$ are the plasmonic eigenvalues and $a_k^{-1}\to 1$ as $k\to\infty$. By \eqref{eqn:isometry} the plasmonic eigenspaces are orthogonal with respect to the scalar product $(\cdot,\cdot)_+$ on $\calH$. There is an eigenvalue $a_k=0$ which corresponds to the kernel of $\calN_-$, that is, the constants in $\calHbar$, and these correspond to the space $E_\infty$.
\end{proof}

\begin{proof}[Proof of Theorem  \ref{thm:regularity}]
Clearly, the correspondence in Proposition \ref{prop:dirichlet-neumann} also holds for $H^s$-plasmons. The operator $\eps\calN_-+\calN_+$ has principal symbol $(\eps-1)|\xi|$. If $\eps\neq 1$ then this is invertible for each $\xi\neq0$, hence the operator $\eps\calN_-+\calN_+$ is invertible then. So by elliptic regularity its kernel consists of smooth functions. These functions are the boundary values of the corresponding plasmons, so these are smooth up to the boundary by elliptic boundary regularity.
\end{proof}
\subsection{Variational principle}
\label{subsec:var princ}

\begin{proof}[Proof of Theorem \ref{thm:euler-lagrange}]

We calculate the first variation of $R_p$. Let $u\in \Hdot^1(\R^n)$ be such that $\int_\Omega |\nabla u|^2 >0$ and let $\eps=R_p(u)$. For $v\in C_0^\infty(\R^n)$, the space of smooth, compactly supported functions, we have
\begin{gather}
\frac{\delta R_p}{\delta u}(v)
= \frac 2 {\int_\Omega |\nabla u_-|^2} \left( \int_{\Omega^c}\nabla u_+\,\nabla v_+ - \eps \int_\Omega \nabla u_-\, \nabla v_- \right) 
\label{eqn:ibp varpr}\\
\qquad\qquad
= \frac 2 {\int_\Omega |\nabla u_-|^2} \left( \int_{\dO} (\partial_n u_+ + \eps \partial_n u_-) v - \int_{\Omega^c} (\Delta u_+) v_+ + \eps \int_{\Omega} (\Delta u_-) v_- \right) 
\label{eqn:ibp varpr2}
\end{gather}
using integration by parts.
If $u$ satisfies \eqref{I1}-\eqref{I3} then it follows that $\frac{\delta R_p}{\delta u}(v) =0$ for all $v\in C_0^\infty(\R^n)$. Now $C_0^\infty(\R^n)$ is dense in
$\Hdot^1(\R^n)$, and for $u$ satisfying in addition \eqref{I4} the right hand side of \eqref{eqn:ibp varpr} depends continuously on $v$, so $u$ is a critical point of $R_p$, and $R_p(u)=\eps>0$ by definition and Theorem \ref{thm:existence+asymp}.

Conversely, assume $u$ is a critical point of $R_p$ with $\eps=R_p(u)>0$.
Using smooth test functions $v$ localized near a point in $\Omega$ or the interior of $\Omega^c$ one sees from \eqref{eqn:ibp varpr}, \eqref{eqn:ibp varpr2} that $\Delta u=0$ on $\R^n\setminus \partial \Omega$, and then using $v$ localized near points on $\dO$ one gets $\eps \partial_n u_- + \partial_n u_+ =0$ at $\dO$. By the Sobolev restriction theorem we have
$u_-=u_+$ at $\dO$. Finally, a harmonic function $u$ on $\Omega^c$ with $\nabla u\in L^2(\Omega^c)$ satisfies \eqref{I4}, so $u$ is a plasmon with eigenvalue $\eps$.
\end{proof}
\ownremark{In terms of Dirichlet-Neumann operators, one can express the variational principle as follows: Let
$$ \Rtilde_p(g) = - \frac{\langle g, \calN_+ g\rangle} {\langle g, \calN_- g\rangle} $$
where $\langle\cdot,\cdot\rangle$ is the scalar product on $L^2(\dO)$.
Then ...
}

\subsection{Scale invariance}
\label{subsec:scale inv}
As opposed to the eigenvalue problem for the Laplacian, the plasmonic eigenvalue problem is scale invariant: If $t>0$ then the domains $\Omega$ and $t\Omega = \{tx:\, x\in\Omega\}$ have the same plasmonic eigenvalues, since an eigenfunction $u$  for $\Omega$ turns into an eigenfunction $u_t(x) := u(x/t)$ for $t\Omega$ with the same $\eps$.

\subsection{The two-dimensional case}
\label{subsec:2d}
The condition \eqref{I4} on the behavior of $u$ at infinity is natural for all dimensions $n\geq 2$ in that it gives unique solvability of the Dirichlet problem on $\Omega^c$ for any boundary data on $\dO$. However, this condition is not suitable for formulating an analogue of the plasmonic eigenvalue problem for $n=2$ because any constant function $u$ will satisfy \eqref{I1}-\eqref{I4} for {\em any} value of $\eps$. On the other hand, adding a constant to a plasmon $u$ with arbitrary plasmonic eigenvalue will yield a plasmon with the same eigenvalue. Therefore, it is natural to consider plasmons as elements of the quotient space $\calH/\{\text{constants}\}$. Since any bounded harmonic function on $\Omega^c\subset\R^2$ is of the form $\,\text{constant}+ O(|x|^{-1})$, this is equivalent to replace the decay condition \eqref{I4} by the condition
\begin{equation}\label{I4 n=2}
u_+(x) = O(|x|^{-1}) \quad \text{as}\quad |x|\to\infty \tag{$\ref{I4}_{n=2}$}  
\end{equation}
Denote the space of plasmons satisfying \eqref{I1}, \eqref{I2} and \eqref{I4 n=2} by $\calH'$. On this space $(\cdot,\cdot)_+$ is a scalar product, and Theorem \ref{thm:existence+asymp} holds, except that now we have the decomposition
$$ \calH' = \bigoplus_{\eps\in\R} E_\eps $$
and there is no space $E_\infty$. Theorems \ref{thm:regularity} and \ref{thm:euler-lagrange} hold as in $n\geq3$, with the same proofs.

In Section \ref{subsec:dir-neu} the definition and discussion of the Dirichlet-Neumann operators $\calN_\pm$ before Proposition \ref{prop:dirichlet-neumann}, and also Proposition \ref{prop:dirichlet-neumann}, hold for $n=2$ (where condition \eqref{I4} is used to define $\calN_+$). However, the operator $\calN_+$ is not invertible if $n=2$. More precisely, we have 
$$\ker\calN_+=\ker\calN_-=\{\text{constants}\} $$
and by self-adjointness both ranges are equal to the orthogonal complement  of the constants (with respect to the standard scalar product \eqref{eq:def scalar product dO}), which we denote by $\calR$.
Therefore, both $\calN_+$ and $\calN_-$ descend to invertible operators
on $C^\infty(\dO)/\{\text{constants}\}$, which may be identified with $\calR$. Then the plasmonic eigenvalues are again the reciprocals of the eigenvalues of $A=-\calN_+^{-1}\calN_-$ considered as operator $\calR\to\calR$ since 
the isomorphism $\calH\longleftrightarrow C^\infty(\dO)$ induces an isomorphism
$\calH/\{\text{constants}\} \longleftrightarrow C^\infty(\dO)/\{\text{constants}\}$. 
The operator $A$ is invertible, hence $\infty$ is not a plasmonic eigenvalue. The proof of Theorem \ref{thm:existence+asymp} carries over if one extends $A$ to $\calHbar$ by continuity and by setting it zero on constants.

There is a subtlety, which is the reason why quotient spaces are more natural in this discussion than their replacements $\calH'$ and $\calR$: The boundary values of harmonic functions on $\Omega^c$ satisfying condition \eqref{I4 n=2} are {\em not} the functions in $\calR$. Rather, they are characterized as follows. There is a unique harmonic function $v$ on $\Omega^c$ having boundary values $v_{|\dO}=0$ and satisfying
$v(x) \sim \frac1{2\pi} \log |x|$ as $|x|\to\infty$. This function may be constructed by choosing any point $y\in\Omega$ and letting $v=G-G'$ where $G(x)=\frac1{2\pi}\log|x-y|$ is the Newton potential centered at $y$ and $G'$ is the unique bounded harmonic function on $\Omega^c$ having the same boundary values on $\dO$ as $G$.  Let $g_0=\partial_\nu v\in C^\infty(\dO)$.
 Then we have:
\begin{quote}
 $g\in C^\infty(\dO)$ is the boundary value of a harmonic function on $\Omega^c$ satisfying \eqref{I4 n=2} if and only if $\langle g,g_0\rangle=0$.
\end{quote}
Necessity follows from Green's formula, and then sufficiency from a dimension argument.
The function $g_0$ is not constant, except when $\Omega$ is a disk. 
This means that the harmonic extension $u$ of $g\in\calR$ will not satisfy the bound \eqref{I4 n=2} in general. In order to obtain a plasmon satisfying \eqref{I4 n=2} one has to add a constant to $u$.
%

\section{Perturbation of plasmonic eigenvalues} \label{secperturb}
In this section we prove Theorem \ref{thm:perturbation}.
\subsection{Analyticity} \label{subsec:analytic}
The analytic dependence of the plasmonic eigenvalues on $h$ is proved by standard methods: First, translate into a problem on a fixed domain, then show that the operators involved depend analytically on $h$, then apply standard perturbation theory.

We provide the details:
In a first step, we translate the problem for $\Omega(h)$ to a ($h$-dependent) problem for the $h$-independent domain $\Omega$. To do this, we extend the vector field $a\bfn$ on $\dO$ to a smooth vector field $V$ on $\R^n$, which vanishes outside some small neighborhood $U$ of $\dO$. This can be done by first extending $a\bfn$ to a neighborhood of $\dO$ and then multiplying by a smooth function on $\R^n$ supported in this neighborhood and equal to one on $\dO$. Then define maps $T_h:\R^n\to\R^n, x\mapsto x+ h V(x)$. For any $h$, $T_h$ is the identity outside $U$, and by compactness of $\dO$ and the inverse function theorem it is a diffeomorphism on $\R^n$ for $h$ sufficiently close to zero. By definition, $T_h$ is a diffeomorphism $\Omega\to\Omega(h)$.

Now the map $T_h$  is used to 'transport'  data related to $\Omega(h)$ back to $\Omega$: Let $\geucl$ be the Euclidean metric on $\R^n$ and $\bf\bfg_h = T_h^* \geucl$ its pull-back. Denote by $\Delta_h$, $\nabla_h$ the Laplace-Beltrami and gradient operators for the Riemannian metric $\bf\bfg_h$ on $\R^n$ and by $\bfn_h$ the normal vector field with respect to $\bf\bfg_h$ on $\dO$. Then the plasmonic eigenvalue problem for $\Omega(h)$ is equivalent to the problem for $\Omega$ with $\Delta, \partial_n$ replaced by $\Delta_h, \partial_{n_h}$. 
More precisely, $u$ solves the former if and only if $T_h^*u:=u\circ T_h$ solves the latter.
Also, in  the discussion of the variational principle one needs to replace $\nabla$ by $\nabla_h$, and in Proposition \ref{prop:dirichlet-neumann} the Dirichlet-Neumann operators need to be replaced by those for the metric $\bf\bfg_h$, denoted $\calN_{\pm, h}$.

We can now state the precise meaning of analyticity in Theorem \ref{thm:perturbation}.
\begin{proposition} \label{prop:analytic}
 Let $\eps\neq 1$ be a finite plasmonic eigenvalue of $\Omega$ with eigenspace $E$. Then there are $h_0>0$ and real analytic functions
$$ \eps^{(i)}: (-h_0,h_0)\to \R,\quad u^{(i)}:(-h_0,h_0) \to \{u\in C(\R^n): \, u_{|\Omegabar}\text{ and }u_{|\R^n\setminus \Omega}\text{ are smooth}\} $$
for $i=1,\dots,\dim E$ such that $u^{(1)}(0),\dots,u^{(\dim E)}(0)$ form a basis of $E$ and the pairs $((T_h^{-1})^*u^{(i)}(h), \eps^{(i)}(h))$ satisfy the plasmonic eigenvalue problem for the domain $\Omega(h)$, for each $h\in(-h_0,h_0)$.
\end{proposition}
The analyticity of $u^{(i)}$ in the given space is to be understood as analyticity of the restrictions in the Hilbert spaces $H^s(\Omega)$ and $H^s(\Omega^c\cap\{|x|<R\})$ for each $R>0$ and $s\in\R$. 
\begin{proof}
We first consider the reduction of the problem given in and after Proposition \ref{prop:dirichlet-neumann}. Thus, we are dealing with an eigenvalue problem for the operator $A_h=-\calN_{+,h}^{-1}\calN_{-,h}$ on $L^2(\dO)$.
The analytic dependence of plasmonic eigenvalues and eigenfunctions, as elements of $L^2(\dO)$, follows from  standard perturbation theory, see \cite{Kat:PTLO}, if we can show that $h\mapsto \calN_{\pm,h}$ are analytic families of operators.  We show this for $\calN_h:=\calN_{-,h}$, the argument for $\calN_{+,h}$ is analogous. By \cite[Theorem VII.4.2]{Kat:PTLO} we only need to show that the associated family of quadratic forms, $\bfq_{h}(g) = \langle g, \calN_{h} g \rangle$, is analytic in $h$, in the sense that the form domain is independent of $h$ and that $h\to \bfq_{h}(g)$ is analytic for each $g$ in this domain. Now since $\calN_{h}$ is an elliptic pseudodifferential operator of order one, the domain of $\bfq_{h}$ is the Sobolev space $H^{1/2}(\dO)$ for each $h$. Also, for each $g\in H^{1/2}(\dO)$ we have (compare \eqref{eq:ibp})
$ \langle g,\calN_{h}g \rangle =  \int_{\Omega} |\nabla_h v_{h} |^2 \, \dvol_h
$
where $v_{h}\in H^1(\Omega)$ is the solution of the Dirichlet problem on $\Omega$ with boundary value $g$ (here we need to impose also the decay condition \eqref{I4} in the case of $\calN_{+,h}$). So we need to prove analyticity of $h\mapsto  \langle g,\calN_{h}g \rangle =  \int_{\Omega} |\nabla_h v_{h} |^2 \, \dvol_h
$.
We first prove analyticity of $h\mapsto v_h$: Since the metric $\bf\bfg_h$ depends analytically on $h$ (in fact, $\bfg_h(x)$ is given by the matrix $(I+h\, DV(x))^t (I+h\,DV(x))$ in standard coordinates), so does $\Delta_h$. Let $G_{h}$ be the inverse of $\Delta_h$ on $\Omega$ with Dirichlet boundary conditions at $\dO$. Then $G_{h}$ is analytic in $h$ by \cite[Theorem VII.1.3]{Kat:PTLO}. We have $v_{h} = v - G_h \Delta_h v$ where $v=v_{0}$ is the unperturbed solution. This shows that $v_h$ is analytic in $h$. Finally,  $\nabla_h v_h$ and $\dvol_h = \sqrt{\det \bf\bfg_h}\,dx$ are analytic in $h$, so we have proved the analytic dependence of $\calN_h$ on $h$.

Now the cited theorems give the existence of $\eps^{(i)}(h)$ and of branches $g^{(i)}(h)\in L^2(\dO)$, depending analytically on $h$. Then since $\eps^{(i)}(0)=\eps\neq 1$ we can choose $h_0$ small enough so that $\eps^{(i)}(h)\neq1$ for all $h\in(-h_0,h_0)$. Then the standard parametrix construction  gives, for any $s\in\N$, a left parametrix  $Q_h$ of 
 $P_h=\eps^{(i)}(h)\calN_{-,h} + \calN_{+,h}$ with error term $R_h=Q_hP_h-I\in\Psi^{-s}(\dO)$, with $Q_h$ depending on $h$ analytically. Then $0 = Q_hP_hg^{(i)}(h) = g^{(i)}(h) + R_hg^{(i)}(h)$, and this shows that $g^{(i)}(h)$ is analytic with values in $H^s(\dO)$. Finally, this implies that $u^{(i)}(h)$, defined as the solution of the Dirichlet problem with respect to $\Delta_h$ with boundary values $g^{(i)}(h)$, depends analytically on $h$ in any Sobolev space.
\end{proof}

\subsection{The perturbation calculation} \label{subsec:perturb calc}
\subsubsection{Procedure}
We first explain our procedure.
Consider \eqref{I1}-\eqref{I4} with all quantities depending on $h$. Differentiate these equations in $h$ and set $h=0$. This yields a system of equations
\begin{align}
\Delta \udot & = 0 \quad\text{in }\R^3\setminus\partial\Omega
\label{eqn1a}\\
\udot_- - \udot_+ & = F_1 \quad\text{on }\partial\Omega
\label{eqn2a}\\
\eps \partial_n \udot_- + \partial_n \udot_+ & = -\epsdot \partial_n u_- + G_1 \quad\text{on } \partial\Omega
\label{eqn3a}
\end{align}
satisfied by $\udot$,
where $F_1,G_1$ are explicit expressions in $u$ and $\eps$, see \eqref{eq:F},\eqref{eq:G}. This is an inhomogeneous version of the system \eqref{I1}-\eqref{I4}. Since the homogeneous system has a nontrivial solution (the given $u$), solvability of the inhomogeneous system implies certain  compatibility conditions for the right hand sides $F_1$, $-\epsdot \partial_n + G_1 $, see Lemma \ref{lem:compat}. This yields a formula for $\epsdot$.
Similarly, differentiating twice in $h$ before setting $h=0$ gives
\begin{align}
\Delta \uddot & = 0 \quad\text{in }\R^3\setminus\partial\Omega
\label{eqn1b}\\
\uddot_-- \uddot_+ & = F_2 \quad\text{on }\partial\Omega
\label{eqn2b}\\
\eps \partial_n \uddot_- + \partial_n \uddot_+ & = -\epsddot \partial_n u_- + G_2 \quad\text{on } \partial\Omega
\label{eqn3b}
\end{align}
with $F_2,G_2$ given in \eqref{eq:F_2}, \eqref{eq:G_2}, and in the same way as before this yields the formula for $\epsddot$.

\subsubsection{Calculation of derivatives at $\dO$}
In order to determine $F_1,G_1,F_2,G_2$ in \eqref{eqn2a}-\eqref{eqn3b}, we need to calculate certain derivatives at the boundary. It is important  to express these in invariant differential geometric terms. For a point $p\in\dO$ we decompose orthogonally  $\R^3=T_p\R^3 = T_p \dO \oplus \Span{\bfn(p)}$ and express all quantities with respect to this decomposition. For example, for a function $u$ we have
$$ \nabla u = \nablad u + (\partial_n u ) \bfn .$$

Denote by $\bfn(q,h)$ the unit normal to $\dO_h$ at $q\in\dO_h$.
Recall that $K$, $H$, $W$ denote the Gauss  curvature, mean curvature and Weingarten map of $\dO$, respectively. Let $\Deltad = \divd \nablad$ be the Laplace-Beltrami operator on $\dO$. The definitions of these and the other quantities is recalled in the proof of Lemma \ref{lem:geom dO in Koords u,v}. 
\begin{proposition}\label{prop:deriv at dO}
Fix $p\in\dO$. Denote by $\bfndot(p)$, $\bfnddot(p)$ the first and second derivative in $h$, evaluated at $h=0$, of the unit normal $\bfn(q_h,h)$ at  $q_h= p + h a(p)\bfn(p) \in\dO_h$.
Then we have at $p$
\begin{align*}
\dot{\boldsymbol{n}}  & =-\nablad a  \\
\ddot{\boldsymbol{n}} & =-W(\nablad a^2) -|\nablad a |^2 \boldsymbol{n} 
\end{align*}
Furthermore, if $x\mapsto u(x,h)$ is harmonic on $\Omega_h$ for each $h$, smooth in $(x,h)$ up to the boundary, then $h\mapsto \nabla u(q_h,h)$ has first and second derivatives at $h=0$
\begin{align*}
\big(\nabla u\big)^{\cdot} &= a(T_1 u + T_1' \partial_n u) && + a(N_1 u + N_1' \dn u)\,\bfn && + \nabla \udot \\
\big(\nabla u\big)^{\cdot\cdot}&= a^2(T_2 u + T_2' \dn u) && + a^2(N_2 u + N_2' \dn u)\,\bfn
&& + 2 (\nabla \udot)^{\cdot} - \nabla \uddot
\end{align*}
where
\begin{align*}
T_1 &=  W \nablad, & T_1' &=  \nablad, \\
N_1 &= -\Deltad, &   N_1' &= 2H, \\
N_2 &= 2 ( -H\Deltad-\divd W \nablad + \nablad H\cdot \nablad),
& N_2' &= -\Deltad+8H^2-2K.
\end{align*}
\end{proposition}
\noindent Here the operators $T_i,T_i', N_i, N_i'$ act on functions defined on $\dO$.

The tangential component of $\big(\nabla u\big)^{\cdot\cdot}$ is not needed, so we don't calculate it.
The term $(\nabla\udot)^{\cdot}$ can be calculated using the same formula as for $(\nabla´u)^{\cdot}$ since $\udot(h)$ is also harmonic for each $h$.
The same conclusion holds if $u(\cdot,h)$ is harmonic on $\R^n\setminus\overline{\Omega_h}$ for each $h$.

The main point in the formulas for the variation of $\nabla u$ is that they involve at most first derivatives in the normal direction. The second and third normal derivatives, which occur in calculating the left hand side, are expressed in terms of tangential derivatives and curvature terms, using the harmonicity of $u$. See Lemma \ref{lemma2}.

\begin{proof}
Fix $p\in \partial\Omega.$ We introduce coordinates in $\R^3$ so that $p=0$
and that the tangent plane to $\partial\Omega$ in $p$,
$T_p\partial\Omega$, is horizontal and $\Omega$ lies below it. Then we can parametrize the
surface $\dO$ near $p$ as a graph:
\begin{equation}\label{1.1}
x(v,w)=(v,w,f(v,w)),
\end{equation}
so that, near $p$, $\partial\Omega=\{x(v,w): v,w\quad \text{near}\quad 0\}$, and $f=0$, $\nabla f=0$ at $(0,0)$.
 We may rotate coordinates so that $f$ has the form
\begin{equation}\label{1.2}
f(v,w)=\frac{\lambda}{2}v^2+\frac{\mu}{2}w^2+{\calO}^3
\end{equation}
for some $\lambda,\mu\in\R,$ where here und in the sequel ${\calO}^k$ denotes any quantity vanishing at least to order $k$ at
$v=w=0.$ Geometrically, $\lambda,\mu$ are the principal curvatures of $\partial\Omega$ at $p$.

Let $\boldsymbol{n}=\boldsymbol{n}(v,w)$ be the upward unit normal:
\begin{equation}
\label{eq:normal}
\boldsymbol{n}(v,w)=\frac{(-f_v(v,w),
-f_w(v,w),1)}{\sqrt{1+f_v^2(v,w)+f_w^2(v,w)}}.
\end{equation}
Expanding this to second order we get
\begin{eqnarray}
\boldsymbol{n}(0,0)&=&(0,0,1),\nonumber\\
\boldsymbol{n}_v(0,0)&=&(-\lambda,0,0),\label{1.3}\\
\boldsymbol{n}_{vv}(0,0)&=&(-f_{vvv},-f_{vvw},-\lambda^2)\nonumber.
\end{eqnarray}
and analogously with $v,w$ interchanged and $\lambda$ replaced by $\mu$.
Here and in the sequel derivatives of $f$ are evaluated at $(0,0)$ if no arguments are given.

We now express geometric quantities of $\dO$ in the coordinates $v,w$.
\begin{lemma}
\label{lem:geom dO in Koords u,v}
For the metric tensor $g$, second fundamental form $II$ and Weingarten map $W$ we have in the coordinates $v,w$
\begin{align}
g &= \Id + \calO^2
\label{1.5} \\
\label{1.6}
II &=W+\calO^2 =
\left( \begin{array}{cc} \displaystyle{\lambda}& 0\\
 0 & \displaystyle{\mu}
\end{array}\right)
+
\left( \begin{array}{cc} \displaystyle{v f_{vvv}+w
f_{vvw}}& \displaystyle{v f_{vvw}+w f_{vww}}\\
\displaystyle{v f_{vvw}+w f_{vww}}& \displaystyle{v
f_{wwv}+w f_{www}}
\end{array}\right)+{\calO}^2
\end{align}
Here $\Id$ is the identity matrix. The derivatives of $f$ are evaluated at the origin.

The intrinsic differential operators in the surface are
\begin{equation}
\label{eq:div nabla Laplace}
\nablad=\nabla+ \calO^2,\ \divd= \div + \calO^1,\ \Deltad = \Delta + \calO^1
\end{equation}
where $\nabla,\div,\Delta$ are the standard (Euclidean) operators in $v,w$.
\end{lemma}
We note the following consequences of the lemma, needed below: The mean curvature is
\[
H=\frac{1}{2}\text{trace}\,W=\frac{1}{2}(\lambda+\mu+v\Delta f_v+w\Delta f_w)+{\calO}^2.
\]
and then simple calculations give
\begin{equation}
\label{1.10}
\divd W \nablad = \lambda \partial_v^2+\mu \partial_w^2 +\nabla (\Delta f)\cdot \nabla +\calO^1
\end{equation}
and
\begin{equation}\label{1.11}
\nablad H= \frac{1}{2}\nabla\Delta f + \calO^1.
\end{equation}

\begin{proof}
For the metric tensor $g$ we have
\begin{eqnarray}
g(v,w)&=& \left(
\begin{array}{cc}
\displaystyle{g_{11}}&\displaystyle{g_{12}}\\
\displaystyle{g_{21}}&\displaystyle{g_{22}}
\end{array}
\right)
 =\left(
\begin{array}{cc}
\displaystyle{\partial_v x\cdot\partial_v x}&\displaystyle{\partial_v x\cdot\partial_w x}\\
\displaystyle{ \partial_w x\cdot\partial_v
x}&\displaystyle{ \partial_w x\cdot\partial_w x}
\end{array}
\right)\nonumber\\
&=&\left(
\begin{array}{cc}
\displaystyle{1+f_v^2(v,w)}&\displaystyle{f_v(v,w) f_w(v,w)}\\
\displaystyle{f_v(v,w) f_w(v,w)}&\displaystyle{1+f_w^2(v,w)}
\end{array}\right) = \Id + \calO^2\nonumber
\end{eqnarray}

The second fundamental form, expressing curvature, is defined by
\[
II=\left(
\begin{array}{cc}
\displaystyle{II_{11}}&\displaystyle{II_{12}}\\
\displaystyle{II_{21}}&\displaystyle{II_{22}}
\end{array}
\right)=\left(\begin{array}{cc}
\displaystyle{\boldsymbol{n}(v,w)\cdot \partial_v^2 x(v,w)}&\displaystyle{\boldsymbol{n}(v,w)\cdot \partial_{vw}^2 x(v,w)}\\
\displaystyle{\boldsymbol{n}(v,w)\cdot \partial_{vw}^2
x(v,w)}&\displaystyle{\boldsymbol{n}(v,w)\cdot \partial_w^2 x(v,w)}
\end{array}\right),
\]
 where
\begin{eqnarray*}
\boldsymbol{n}(v,w)\cdot \partial_v^2
x(v,w)&=&\frac{(-f_v(v,w),-f_w(v,w),1)}{\sqrt{1+f_v^2(v,w)+f^2_w(v,w)}}\cdot
(0,0,f_{vv}(v,w))=\frac{f_{vv}(v,w)}{\sqrt{1+f_v^2+f_w^2}}\\
&=&f_{vv}(v,w)+{\calO}^2,\\
\boldsymbol{n}(v,w)\cdot \partial_{vw}^2 x(v,w)&=&f_{vw}(v,w)+{\calO}^2,\\
\boldsymbol{n}(v,w)\cdot \partial_w^2 x(v,w)&=&f_{ww}(v,w)+{\calO}^2.
\end{eqnarray*}
Expanding to first order in Taylor series gives the formula for $II$ in \eqref{1.6}.

The Weingarten map (or shape operator), again expressing curvature, has components
\[
W_i^j=\sum\limits_{k=1}^2 II_{ik}g^{kj},\quad i,j=1,2,
\]
where $\Big(g^{kj}\Big)_{k,j=1,2}$ is the inverse matrix of $g$. From (\ref{1.5}) it follows that $g^{-1}=\Id+{\calO}^2$ and hence
$W=II\cdot g^{-1}=II+{\calO}^2$, so \eqref{1.6} is proved.

The gradient of $u$ in the surface is
\begin{equation*} 
\nablad u=\left(
\begin{array}{cc}
\displaystyle{g^{11}}&\displaystyle{g^{12}}\\
\displaystyle{g^{21}}&\displaystyle{g^{22}}
\end{array}
\right)\left(
\begin{array}{c}
u_v\\
u_w
\end{array}
\right)=\Big( \Id+{\calO}^2\Big)\left( \begin{array}{c}
u_v\\
u_w
\end{array}\right)=\left( \begin{array}{c}
u_v\\
u_w
\end{array}\right)+{\calO}^2.
\end{equation*}

Next, $g=\Id+\calO^2$ implies $\sqrt{\det g}=1+\calO^2$, so for a vector field $Z=(Z_v, Z_w)$  the divergence in the surface is
\begin{eqnarray*}
\divd Z&=&\frac{1}{\sqrt{\text{det}\, g}}\Big(\partial_v(\sqrt{\text{det}\,g}Z_v)+\partial_w(\sqrt{\text{det}\,g}Z_w)\Big)\\
&=&\partial_v Z_v+\partial_w Z_w+{\calO}^1,
\end{eqnarray*}
and this gives also the result for the Laplace-Beltrami operator, $\Deltad=\divd\nablad$.
\end{proof}

We now prove the formulas for $\bfndot,\bfnddot$ in Proposition \ref{prop:deriv at dO}. We parametrize $\partial\Omega_h$ by
$$Y(v,w,h)=x(v,w)+ha(v,w)\boldsymbol{n}(v,w)$$
with $x(v,w)$ as in \eqref{1.1}. For $p=0$ we have $q_h=Y(0,0,h)$, and $N(h)=Y_v(0,0,h)\times Y_w(0,0,h)$ is a non-normalized normal to $\dO_h$ at this point.
From \eqref{1.3} we have 
$Y_v(0,0,h)=\big(1-a\lambda h,0,h a_v\big) $ and analogously
$Y_w(0,0,h)=\big(0,1-a\mu h,h a_w\big)$, so (with $a$ and its derivatives taken at $(0,0)$)
\begin{align*}
N(h) &= (-a_v h+aa_v\mu h^2,\, -a_w h+aa_w\lambda  h^2,
1-a(\mu+\lambda)h+a^2\lambda\mu h^2) \\
&=A+Bh+Ch^2
\end{align*}
where $A=(0,0,1),\ B=(-a_v,-a_w,-2H a),\ C=(a a_v\mu,aa_w\lambda,a^2 K)$.
Expanding $\bfn(q_h,h)=\frac{N(h)}{||N(h)||}$ in Taylor series results in
$\bfn(q_h,h) =$
$$
A+\Big(B-A(A\cdot B)\Big)h+\left( C-B(A\cdot B)-A\Big(\frac{||B||^2}{2}+A\cdot C-\frac{3}{2}(A\cdot B)^2\Big)\right)h^2+O(h^3).
$$
The coefficient of $h$ is $(-a_v,-a_w,0)=-\nablad a $ by \eqref{eq:div nabla Laplace}, and the coefficient of $h^2/2$ is
$(-2aa_v\lambda,-2a a_w\mu,-{|\nabla a|^2}) =  -W(\nablad a^2)-|\nablad a|^2\boldsymbol{n}$ by \eqref{1.6} and \eqref{eq:div nabla Laplace}, which proves the first part of Proposition \ref{prop:deriv at dO}.

To prove the second part of Proposition \ref{prop:deriv at dO} it is convenient to use coordinates normal with respect to $\dO$, near $p=0$. That is, we parametrize a point $X=(X_1,X_2,X_3)$ near $p$ by coordinates $(v,w,s)$ where $s$ is the signed distance from $X$ to $\dO$ and $x(v,w)$ is the point on $\dO$ closest to $X$. That is,
\[
X(v,w,s)=x(v,w)+s\boldsymbol{n}(v,w).
\]
\begin{lemma}\label{lemma1}
The derivatives in the coordinates $X_1, X_2, X_3$ and $v,w,s$ at the point $p=0$ are related as follows:
\begin{eqnarray}
&&\partial_{X_1}=\partial_v,\quad \partial_{X_2}=\partial_w,\quad
\partial_{X_3}=\partial_s;\nonumber\\
&&\partial^2_{X_1}=\partial_v^2-\lambda\partial_s,\quad
\partial_{X_2}^2=\partial_w^2-\mu\partial_s,\quad
\partial^2_{X_3}=\partial_s^2;\nonumber\\
&&\partial_{X_1}\partial_{X_3}=\partial_v\partial_s+\lambda\partial_v,\quad
\partial_{X_2}\partial_{X_3}=\partial_w\partial_s+\mu\partial_w;\label{1.4}\\
&&\partial^2_{X_1}\partial_{X_3}=\partial^2_v\partial_s-\lambda
\partial_s^2+2\lambda\partial^2_v-\lambda^2\partial_s+f_{vvv}\cdot
\partial_v+f_{vvw}\cdot \partial_w;\nonumber\\
&&\partial^2_{X_2}\partial_{X_3}=\partial^2_w\partial_s-\mu
\partial_s^2+2\mu\partial^2_w-\mu^2\partial_s+f_{www}\cdot
\partial_w+f_{wwv}\cdot \partial_v\nonumber.
\end{eqnarray}
\end{lemma}

\begin{proof}
By differentiating and using the chain rule repeatedly, one  gets
\begin{eqnarray*}
&&\partial_s=\frac{\partial X}{\partial x}\partial_X = \boldsymbol{n}\cdot \partial_X, \quad \partial_s^2=\boldsymbol{n}\cdot\boldsymbol{n}\cdot
\partial_X^2;\\
&&\partial_v=X_v\cdot \partial_X,\quad
\partial_v^2=\partial_v(X_v\cdot \partial_X)=X_v\cdot X_v\cdot
\partial_X^2+X_{vv}\cdot \partial_X;\\
&&\partial_v \partial_s=\boldsymbol{n}\cdot X_v\cdot
\partial_X^2+\boldsymbol{n}_v\cdot \partial_X,\\
&&\partial_v^2\partial_s=\boldsymbol{n}\cdot
X_v\cdot X_v\cdot
\partial_X^3+\boldsymbol{n}\cdot X_{vv}\cdot
\partial^2_X+2\boldsymbol{n}_v\cdot X_v\cdot
\partial_X^2+\boldsymbol{n}_{vv}\cdot\partial_X,
\end{eqnarray*}
 and analogous formulas for the $w$-derivatives\footnote{For vectors $V,W$ we write $V\cdot W\cdot\partial_X^2 =\sum\limits_{i,j=1}^3 V_i W_j \partial_{X_i}\partial_{X_j}$, for short, and similar for third derivatives in $X$.}.

We now evaluate all quantities at $p=0$, i.e.
$v=w=s=0$. Using \eqref{1.3} and
$
X_{v}(0,0,0)=(1,0,0),\ X_{vv}(0,0,0)=(0,0,\lambda)
$
we get at $p$
\begin{eqnarray*}
&&\partial_s=\partial_{X_3},\quad
\partial_s^2=\partial_{X_3}^2;\\
&&\partial_v=\partial_{X_1},\quad
\partial_v^2=\partial^2_{X_1}+\lambda\partial_{X_3};\\
&&\partial_v
\partial_s=\partial_{X_1}\partial_{X_3}-\lambda\partial_{X_1},\\
&&\partial_v^2 \partial_s=\partial^2_{X_1}\partial_{X_3}+\lambda
\partial^2_{X_3}-2\lambda\partial^2_{X_1}-f_{vvv}\partial_{X_1}-f_{vvw}\partial_{X_2}-\lambda^2
\partial_{X_3}.
\end{eqnarray*}

Analogous formulas hold with $v$ replaced by $w$, $\lambda$ by
$\mu$ and with $\partial_{X_1},\partial_{X_2}$ interchanged.
 Solving these equations  for the derivatives   $\partial_{X_i}$ one gets the desired relations.
 \end{proof}

\begin{lemma}\label{lemma2}
Let $u$ be harmonic on $\Omega$, smooth up to $\dO$. Then
\begin{eqnarray}
(\bfn\cdot\nabla)^2 u&=&N_1 u + N_1' \dn u\label{1.12},\\
(\bfn\cdot\nabla)^3 u&=&N_2 u + N_2' \dn u \label{1.13}
\end{eqnarray}
at $p=0$, with $N_i, N_i'$ as given in Proposition \ref{prop:deriv at dO}.
\end{lemma}
Here $\left[(\bfn\cdot\nabla)^2 u\right](p)$ is understood as $\left(\bfn(p)\nabla\right)^2 u$, evaluated at $p$; that is, $\bfn$ is not differentiated\footnote{Of course it would not make sense to differentiate it since it is only defined at $\dO$. But this is in contrast to the consideration of $(\dn u)^{\cdot}$ below, where the normal depends on $h$, hence must be differentiated.}. In the coordinates above, we have at $p=0$
$$ (\bfn\cdot\nabla)^2 u=\partial^2_{X_3}u,\quad (\bfn\cdot\nabla)^3 u =
\partial^3_{X_3}u $$
\begin{proof}
A short calculation, using (\ref{1.4}), \eqref{eq:div nabla Laplace} and $\partial_s=\partial_{X_3}=\partial_{n}$, gives, at $0$,
\begin{eqnarray*}
\partial^2_{X_3}u &=&-(\partial^2_{X_1}+\partial^2_{X_2})u\\
&=&-(\partial_v^2-\lambda \partial_s)u-(\partial_w^2-\mu\partial_s)u\\
&=&-(\partial_v^2+\partial_w^2)u+(\lambda+\mu)\partial_s u\\
&=&-\Deltad u+2H\partial_{n} u.
\end{eqnarray*}
Next we evaluate $\partial^3_{X_3}u$ at $0$. By the last formulas in (\ref{1.4}) we get
\begin{eqnarray*}
\partial^3_{X_3}u &=& -(\partial^2_{X_1}+\partial^2_{X_2})\partial_{X_3}u\\
&=&-(\partial_v^2 \partial_s+\partial_w^2 \partial_s)u+(\lambda+\mu)\partial^2_s u\\
&-&2(\lambda\partial_v^2+\mu\partial_w^2)u+(\lambda^2+\mu^2)\partial_s u\\
&-&(\underbrace{f_{vvv}\partial_v+f_{www}\partial_w+f_{vvw}\partial_w+f_{wwv}\partial_v}_{\nabla(\Delta f)\cdot\nabla})u.
\end{eqnarray*}

Using \eqref{eq:div nabla Laplace}, (\ref{1.10}), (\ref{1.11}) and then  $\partial_s=\partial_{n}$ (everywhere),  $\partial^2_s u =\partial^2_{X_3}u=(-\Deltad+2H\partial_{n})u$ (at $p$ only) we obtain
\begin{align*}
\partial^3_{X_3}u &=-\Deltad\partial_s u+2H\partial_s^2 u+(4H^2-2K)\partial_s u
-2\divd(W(\nablad u))+2\nablad H\cdot \nablad u\\
&= \Big(-\Deltad+8H^2-2K\Big)\partial_{n}u+
2\Big(-H\Deltad-\divd(W(\nablad))+\nablad H\cdot\nablad\Big)u
\end{align*}
as claimed.
\end{proof}

Now we prove the second part of Proposition \ref{prop:deriv at dO}.
We have $\frac{d}{dh} \left[(\nabla u)(q_h,h)\right] = (\dot{q}_h\cdot\nabla)(\nabla u)(q_h,h) + (\nabla \udot)(q_h,h)$. With $q_h=p+ha(p)\bfn(p)$ we get $\dot{q}_h=a(p)\bfn(p)$. At $p=0$, $\bfn(p)\cdot\nabla = \partial_{X_3}$.
So we get, using (\ref{1.4}) and (\ref{1.12}),  with everything evaluated at $p=0$, $h=0$,
\begin{eqnarray*}
\big(\nabla u\big)^{\cdot}&=&a\big(\boldsymbol{n}\cdot \nabla\big)\nabla u+\nabla \dot{u}=
a\left(\begin{array}{c}\partial_{X_3}\partial_{X_1}u\\
\partial_{X_3}\partial_{X_2}u\\
\partial^2_{X_3}u\end{array}\right)+\nabla \dot{u}\\
&=&a\left\{\left(\begin{array}{c}
\partial_v \partial_n u\\
\partial_w \partial_n u\\
0
\end{array} \right)+\left(\begin{array}{c}
\lambda \partial_v u\\
\mu \partial_w u\\
0
\end{array}\right)+\left(\begin{array}{c}
0\\
0\\
N_1 u + N_1' \dn u
\end{array}\right)\right\}+\nabla\dot{u}\\
\end{eqnarray*}
The first term is $a\nablad \partial_n u$ by \eqref{eq:div nabla Laplace}, and the second is $aW(\nablad u)$ by \eqref{1.6}, \eqref{eq:div nabla Laplace}.

Finally,
$$ \big(\nabla u\big)^{\cdot\cdot}=a^2 \big(\boldsymbol{n}\cdot \nabla\big)^2 \nabla u +2a\big(\boldsymbol{n}\cdot \nabla\big)\nabla \dot{u}+\nabla\ddot{u} $$
and the normal component of the first term is
$a^2(\partial_{X_3}^3 u) \bfn = a^2(N_2 u + N_2' \dn u)\bfn$ by (\ref{1.13})
while the last two terms can be rewritten as $2(\nabla\udot)^{\cdot} - \nabla \uddot$.
This finishes the proof of Proposition \ref{prop:deriv at dO}.
\end{proof}

\begin{lemma}\label{lem:partial n u dot}
 Let $u$ be as in Proposition \ref{prop:deriv at dO}.
The derivatives of $h\mapsto u(q_h,h)$ at  $h=0$ are
\begin{equation}
\label{eq:udot uddot}
(u)^{\cdot} = a\,\dn u + \udot,\qquad (u)^{\cdot\cdot} = -a^2\Deltad u + 2a^2 H \dn u + 2a\, \dn \udot + \uddot
\end{equation}
The derivatives of $h\mapsto(\partial_n u)(q_h,h)$ at  $h=0$ are
\begin{align}\label{2.3}
\big(\partial_n u\big)^{\cdot} &= P_1(u) + Q_1(\partial_n u) + \partial_n\dot{u}\\
\label{2.7}
\big(\partial_n u\big)^{\cdot\cdot} &=P_2(u)+Q_2\big(\partial_n u\big)+2P_1(\dot{u})+2Q_1 (\partial_n \dot{u}) +\partial_n \ddot{u}
\end{align}
 with
 \begin{eqnarray*}
 P_1&=&-\divd a\nablad,\\
 Q_1&=&2aH,\\
 P_2&=&-2\Big(a^2 H \Deltad +
\divd (a^2 W) \nablad - a^2(\nablad H)\cdot\nablad \Big),\\
 Q_2&=&-\divd a^2\nablad + a^2(8H^2-2K)-|\nablad a|^2,\\
 \end{eqnarray*}
\end{lemma}
 \begin{proof}
The first equation of \eqref{eq:udot uddot} is obvious. The second one follows from
$\left(u(q_h,h)\right)^{\cdot\cdot}= a^2 (\bfn\cdot\nabla)^2 u  + 2a(\bfn\cdot\nabla) \udot  + \uddot$
and \eqref{1.12}.

Equation \eqref{2.3} follows from $ \big(\partial_n u\big)^{\cdot}=\big(\boldsymbol{n}(q_h,h)\cdot\nabla u(q_h,h)\big)^{\cdot}=\dot{\boldsymbol{n}}\cdot\nabla u+
 \boldsymbol{n}\cdot \big(\nabla u\big)^{\cdot},
 $
 and  Proposition \ref{prop:deriv at dO}: The terms acting on $u$ are $$-(\nablad a)\cdot\nablad + a N_1 = -(\nablad a)\cdot\nablad - a\Deltad= -\divd a \nablad=P_1$$
and there is one term acting on $\dn u$, which is $a N_1'=Q_1$.

For \eqref{2.7} we write
$
 \big(\partial_n u\big)^{\cdot\cdot}=\big( \boldsymbol{n}(q_h,h)\cdot \nabla u(q_h,h)\big)^{\cdot\cdot} = \ddot{\boldsymbol{n}}\cdot \nabla u+2\dot{\boldsymbol{n}}\cdot \big(\nabla u\big)^{\cdot}+\boldsymbol{n}\cdot \big(\nabla u\big)^{\cdot\cdot}
$, use Proposition \ref{prop:deriv at dO} and sort by terms involving $u$, $\dn u$, $\udot$, $\dn\udot$. The terms acting on $u$ are
\begin{align*}
 & -W(\nablad a^2)\cdot\nablad
-2(\nablad a) a T_1 + a^2 N_2  \\
=& -(\nablad a^2) \cdot W \nablad - (\nablad a^2)\cdot W\nablad +a^2 N_2\qquad\text{ (since $W$ is selfadjoint)} \\
=& -2 \left[ (\nablad a^2)\cdot W\nablad + a^2\divd W \nablad + a^2H\Deltad  - a^2(\nablad H)\nablad \right]= P_2
\end{align*}
The terms acting on $\dn u$ are
\begin{align*}
& -|\nablad a|^2 - 2(\nablad a)aT_1'+ a^2 N_2' \\
=&
-|\nablad a|^2 - (\nablad a^2)\cdot\nablad - a^2\Deltad + a^2(8H^2-2K) = Q_2
\end{align*}
The terms involving $\udot$ and $\dn\udot$ arise from one of the $h$-derivatives in $\big( \boldsymbol{n}(q_h,h)\cdot \nabla u(q_h,h)\big)^{\cdot\cdot}$ falling on the second argument of $u$, hence are twice the corresponding terms in $(\dn \udot)^{\cdot}$.
\end{proof}

\subsubsection{End of proof of the perturbation formulas}

Now let $u(\cdot,h)$, $\eps(h)$ be a solution of \eqref{I1}-\eqref{I4} for $\Omega_h$, for each $h$, smoothly depending on all variables (up to the boundary of $\Omega_h$ from either side).

We first prove \eqref{eqn1a}-\eqref{eqn3b} and calculate $F_1,G_1,F_2,G_2$.
If $x\not\in\dO$ then $x\not\in \partial\Omega_h$ for sufficiently small $h$, and \eqref{eqn1a},\eqref{eqn1b} follow. The boundary conditions are
\begin{align}
\label{eq:u(h) eqn}
u_-(q_h,h) - u_+(q_h,h) &=0 \\
\label{eq:dn u(h) eqn}
\eps(h)\partial_n u_- (q_h,h) + \partial_n u_+ (q_h,h) &=0
\end{align}
where $q_h= p +ha(p) \bfn(p)$, for each $p\in\dO$ and $h$.
We differentiate \eqref{eq:u(h) eqn}  and use the first equation in \eqref{eq:udot uddot}
for $u_-$ and $u_+$, at $h=0$. This gives $\udot_--\udot_+=-a(\dn u_- - \dn u_+)$. Using \eqref{eq:dn u(h) eqn} at $h=0$ we obtain \eqref{eqn2a} with
\begin{equation}
\label{eq:F}
 F_1 = -  (\eps+1) a\dn u_-
\end{equation}
Next, we differentiate \eqref{eq:dn u(h) eqn} at $h=0$ and use \eqref{2.3} for $u_-$ and $u_+$. We obtain
\begin{align*}
0 &= \epsdot \dn u_- + \eps (\dn u_-)^{\cdot} + (\dn u_+)^{\cdot}  \\
&=
\epsdot \dn u_- + P_1 (\eps u_- + u_+) + Q_1 (\eps \dn u_- + \dn u_+)
+ \eps \dn \udot_- + \dn \udot_+ \\
&= \epsdot \dn u_-  + (\eps+1) P_1 u_- + \eps \dn \udot_- + \dn \udot_+
\end{align*}
using \eqref{eq:u(h) eqn}, \eqref{eq:dn u(h) eqn} at $h=0$, hence \eqref{eqn3a} with
\begin{equation}
\label{eq:G}
 G_1 = -(\eps + 1) P_1 u_-
\end{equation}

Differentiating \eqref{eq:u(h) eqn} twice and using \eqref{eq:udot uddot} for $u_-$ and $u_+$ we obtain
\begin{align*}
0
&= \left(u_-\right)^{\cdot\cdot} - \left(u_+\right)^{\cdot\cdot} \\
&= - a^2 \Deltad (u_- - u_+) + a Q_1 (\dn u_- - \dn u_+)  + 2a (\dn \udot_- - \dn \udot_+) + \uddot_- - \uddot_+
\end{align*}
By \eqref{eq:u(h) eqn}, \eqref{eq:dn u(h) eqn} at $h=0$ the first term vanishes and the second equals $(\eps+1) a Q_1 \dn u_-$. By \eqref{eqn3a} the third term is $(\eps+1)2a \dn \udot_- + \epsdot 2a\dn u_- - 2a G_1$.
Rearranging and using \eqref{eq:G} we obtain \eqref{eqn2b} with
\begin{equation}
\label{eq:F_2}
 F_2 = -(\eps + 1)D - 2a\epsdot \dn u_-
\end{equation}
$$ D = a\big( 2 P_1 u_- +  Q_1  \dn u_- + 2\dn \udot_-  \big)$$

Finally, we differentiate \eqref{eq:dn u(h) eqn} twice and get, using \eqref{2.3} and \eqref{2.7},
\begin{align*}
0
&=  \ddot{\varepsilon}\partial_{n} u_{-}+2\dot{\varepsilon}(\partial_{n} u_{-})^\cdot+\varepsilon (\partial_{n} u_{-})^
 {\cdot\cdot}+(\partial_{n} u_{+})^{\cdot\cdot}  \\
&= \epsddot \dn u_- + 2\epsdot (P_1 u_- + Q_1 \dn u_- + \dn \udot_-) \\
&+ P_2 (\eps u_- + u_+) + Q_2 (\eps \dn u_- + \dn u_+)
+ 2 P_1 (\eps \udot_- + \udot_+) + 2 Q_1 (\eps \dn \udot_- + \dn\udot_+) \\
&+ \eps \dn \uddot_- + \dn \uddot_+  \\
\intertext{The terms in the third line of this equation are, using \eqref{eq:u(h) eqn}, \eqref{eq:dn u(h) eqn} at $h=0$ and \eqref{eqn2a}, \eqref{eqn3a},}
& (\eps+1) P_2 u_- + 0 + 2P_1 ((\eps+1) \udot_- - F_1) +
2 Q_1 (-\epsdot \dn u_- + G_1) \\
\intertext{
Using $F_1 = -  (\eps+1) a\dn u_-$, $G_1 = -(\eps + 1) P_1 u_-$ we obtain }
0 &= \epsddot \dn u_- + 2\epsdot A + (\eps+1)B +  \eps \dn \uddot_- + \dn \uddot_+
\end{align*}
with
\begin{align*}
A &= P_1 u_- + \dn \udot_- \\
B &= (P_2  - 2 Q_1 P_1) u_- + 2 P_1 a \dn u_- + 2 P_1 \udot_-
\end{align*}
Rearranging, we get \eqref{eqn3b} with
\begin{equation}
\label{eq:G_2}
G_2 = -2\epsdot A - (\eps+1) B
\end{equation}
Before proceeding we simplify $B$. Note that
\begin{eqnarray*}
P_2 - 2Q_1 P_1
&=&-2a^2 H\Deltad+2a^2 \big(\nablad H\big)\cdot \nablad-2\divd (a^2 W) \nablad +4aH\big( a \Deltad+(\nablad a)\cdot \nablad\big)\\
&=&-2\divd\big(a^2 W\big)\nablad +2\divd\big(a^2 H\big)\nablad\\
&=&-2\divd\big(a^2 W_0)\nablad
\end{eqnarray*}
where $W_0=W-HI$. Thus
\[
B=-2\divd\big(a^2 W_0\big)\nablad u_- + 2 P_1\big( a \dn u_- + \dot{u}_-\big)
\]

We have proved \eqref{eq:udot}. In order to calculate $\epsdot$, $\epsddot$, we use the following lemma.

\begin{lemma}[Compatibility condition]
\label{lem:compat}
Let $\eps$ be a plasmonic eigenvalue with eigenspace $E$. Given functions $F,\Gtilde$ on $\dO$, the inhomogeneous system
\begin{align*}
\Delta v & = 0 \quad\text{in }\R^3\setminus\partial\Omega
\\
v_- - v_+ & = F \quad\text{on }\partial\Omega
\\
\eps \partial_n v_- + \partial_n v_+ & = \Gtilde \quad\text{on } \partial\Omega
\end{align*}
has a solution $v$ satisfying the decay condition \eqref{I4} iff
$$ \langle \Gtilde,w_{-} \rangle = \eps \langle F,\partial_n w_{-} \rangle$$
for all $w\in E$. Here $\langle\cdot,\cdot\rangle$ is the scalar product on $L^2(\dO)$.
\end{lemma}
\begin{proof}
We reformulate this in terms of Dirichlet-Neumann operators, compare Proposition \ref{prop:dirichlet-neumann}.
Let $g=v_{+}$ at $\dO$. Then the second equation says $v_- = F+g$, and the first and third that $\eps \calN_- (F+g) + \calN_+ g = \Gtilde$. Therefore, a solution $v$ exists if and only if the equation
$$ (\eps \calN_- + \calN_+) g = \Gtilde-\eps \calN_- F $$
has a solution $g$. Since the operator $\eps \calN_- + \calN_+$ is selfadjoint, this is equivalent to $\langle \Gtilde-\eps\calN_- F,w_{-}\rangle = 0$ for all $w\in E$. This is equivalent to the claim.
\end{proof}

We apply the lemma with $v=\udot$ and $w=u$ and $F=F_1$, see \eqref{eq:F}, $\Gtilde=-\epsdot \dn u_- + G_1$, see \eqref{eq:G}. Using the assumption $\langle u_- ,\dn u_- \rangle = \|u_-\|_-^2 = 1$ we get
$$ \epsdot = \langle G_1,u_- \rangle - \eps \langle F_1, \dn u_- \rangle
$$
Integrating by parts $\langle G_1,u_- \rangle = (\eps+1) \langle \divd a \nablad u_-, u_- \rangle = -(\eps+1) \langle a\nablad u_-,\nablad u_-\rangle$ we obtain \eqref{eq:epsdot}.

Now we apply the lemma with $v=\uddot$ and $w=u$ and $F=F_2$, see \eqref{eq:F_2}, $\Gtilde=-\epsddot \dn u_- + G_2$, see \eqref{eq:G_2}. We obtain
$$ \epsddot = \langle G_2,u_- \rangle - \eps \langle F_2, \dn u_- \rangle
$$
Now we sort by like terms. For example, the terms involving $u_-,u_-$ are
$$ -2\epsdot\langle P_1 u_-,u_- \rangle - (\eps+1)(-2) \langle \divd a^2 W_0 \nablad u_-,u_- \rangle $$
where $P_1 = -\divd a\nablad$. Integrating by parts yields the first line in
\eqref{eq:epsddot},
and the other terms are obtained in a similar way.

\ownremark
{One may try to simplify this by choosing a particular $\udot$. For example, one may demand that $u(h)$ is normalized by $\| u(h)\|_-=1$ for all $h$. Differentiating
$ 1 = \|u(h)\|_-^2 = \int_{\Omega_h} |\nabla u(x,h)|^2\,dx $
at $h=0$ yields
$ 0 = 2\int_\Omega \nabla u\cdot\nabla \udot\,dx + \int_{\dO} a |\nabla  u_-|^2\,dx 
= 2 \langle u_-,\dn \udot_- \rangle + \scpr{\nablad u_-}{a}{\nablad u_-} +
\scpr{\dn u_-}{a}{\dn u_-}
$
which can be used to remove the $\langle u_-,\dn \udot_- \rangle$ term, but it does not seem worth it (the other terms are just added, don't cancel anything).
}

\subsection{Shape derivative of Dirichlet-Neumann operator}
As a side result of our calculations above we can calculate the derivative of the Dirichlet-Neumann operator $\calN:=\calN_-$ when changing the domain as in \eqref{eq:def Omega_h}. Since the domain $\dO_h$ varies, we have to state precisely what we mean by this. Recall from Section \ref{subsec:analytic} the diffeomorphism $T_h:\Omega\to\Omega_h,\ x\mapsto x+h V(x)$ where $V=a\bfn$ on $\dO$. This induces an isomorphism $T_h^*:\Cinfty(\dO_h) \to \Cinfty(\dO),\ g\mapsto T_h^*g = g\circ T_h$. Then we consider
$$\calN_h := T_h^* \circ \calN_{\Omega_h} \circ (T_h^*)^{-1} : \Cinfty(\dO) \to \Cinfty(\dO)$$
where $\calN_{\Omega_h}:\Cinfty(\dO_h) \to \Cinfty(\dO_h)$ is the Dirichlet-Neumann operator for $\Omega_h$. The operator $\calN_h$ is the same as the one introduced in Section \ref{subsec:analytic}. Explicitly, it is given as follows: Let $g\in\Cinfty(\dO)$. Let $u(\cdot,h)$ be the solution of
\begin{align}
\Delta u(\cdot,h) &= 0 \quad\text{ in }\Omega_h \label{eq:DN1}\\
u(T_h(x),h) &= g(x) \quad\text{ for all }x\in \dO \label{eq:DN2}\\
\text{and then set } (\calN_h g)(x) &= (\dn u) (T_h(x),h). \label{eq:DN3}
\end{align}
\begin{theorem} \label{thm:DN-perturb}
The derivative at $h=0$ of the familiy of Dirichlet-Neumann operators defined above is given by
\begin{equation}
\label{eq:DN perturb}
\calNdot = -\divd a \nabla + 2aH \calN - \calN a \calN
\end{equation}
\end{theorem}
\begin{proof}
Differentiate \eqref{eq:DN1}, \eqref{eq:DN2}, \eqref{eq:DN3} at $h=0$ and use \eqref{2.3} (where $q_h=T_h(p)$) for the right side of \eqref{eq:DN3}. This gives
\begin{align*}
\Delta \udot &= 0 \quad\text{ in }\Omega\\
a \dn u + \udot &= 0 \quad\text{ at } \dO\\
\calNdot g &= -\divd a \nablad u + 2aH \dn u + \dn \udot
\end{align*}
Rewrite the second equation as $\udot = -a\dn u$ at $\dO$, then together with the first equation this gives $\dn \udot = \calN (-a\dn u) = -\calN (a \calN g)$, and then the third equation gives the claim.
\end{proof}
For the exterior Dirichlet-Neumann operator $\calN_+$ one has the same formula, replacing $\calN$ by $\calN_+$ everywhere. Same proof.

Using the second derivative formulas in Section \ref{subsec:perturb calc} one could also derive a formula for the second derivative of $\calN_h$ in $h$.

Note that the principal symbol of the first and third term in \eqref{eq:DN perturb} is $a|\xi|^2$, hence $\calNdot$ is a pseudodifferential operator of order at most one, which is to be expected since each $\calN_h$ is a pseudodifferential operator of order one.

\subsection{Splitting at first order}
We now prove the last statement of Theorem \ref{thm:perturbation}, which gives also another (though similar) proof of the formula for $\epsdot$. We use the characterization in terms of Dirichlet-Neumann operators, Proposition \ref{prop:dirichlet-neumann}, where we use the operators $\calN_h$ transplanted from $\Omega_h$ to $\Omega$ as in the previous section. If $g(x,h) = u(T_h(x),h)$ for $x\in\dO$ then $(\eps(h)\calN_{-,h} + \calN_{+,h}) g(h) = 0$ for each $h$.
Differentiating in $h$ at $h=0$ gives
$$ (\epsdot\calN_- + \eps\calNdot_- + \calNdot_+)g + (\eps \calN_- + \calN_+)\gdot=0$$
We take the scalar product of this with any $g'\in E_\partial :=\{u_{|\dO}:u\in E\}$. Using selfadjointness of $\calN_\pm$ and $(\eps \calN_- + \calN_+)g'=0$ we get
$$\langle (\epsdot\calN_- + \eps\calNdot_- + \calNdot_+)g, g' \rangle = 0
\quad \forall g'\in E_\partial $$
or equivalently
$$ - \langle (\eps\calNdot_- + \calNdot_+)g, g' \rangle = \epsdot (g,g')_- $$
A simple calculation using \eqref{eq:DN perturb} for $\calN=\calN_\pm$ reveals that the left hand side equals $q_1(u,u')$, where $u'$ is the harmonic function on $\Omega$ with boundary values $g'$, and the claim follows.

\section{Examples}\label{sec:examples}
\subsection{Planar disk}
\label{subsec:disk}
For the disk $\Omega=\{x\in\R^2:\, |x| < 1\}$ the only plasmonic eigenvalue is $\eps=1$, and any harmonic function in $\Omega$ 
whose value and normal derivative at $\dO$ are defined (e.g.\ any function in  $H^s(\Omega)$ with $s>3/2$) is a plasmon.
This follows from the fact that if $u_-$ is harmonic in $\Omega$ then its Kelvin transform $u_+ (x) = u_-(\frac x{|x|^2})$ is harmonic in $|x|>1$. Clearly, $u_-$ and $u_+$ have the same boundary values at $|x|=1$ and their normal derivatives add to zero, so $\eps=1$. The orthogonality in Theorem \ref{thm:existence+asymp} then shows that there can be no other solutions. Choosing $u$ with $u_{|\dO}$ non-smooth we see that the regularity statement in Theorem \ref{thm:regularity} does not hold.

\subsection{Perturbed ball} \label{subsec:ex ball}
For the unit ball, $\Omega=\{x\in\R^3:\, |x| < 1\}$, the solutions of the plasmonic eigenvalue problem are given by
$$ u_{-,k} (r,\omega) = r^k Y_k (\omega),\ u_{+,k} (r,\omega) = r^{-k-1} Y_k (\omega),\ \epsilon_k = \frac{k+1}k $$
for $k=0,1,\dots$.
Here $r>0$, $\omega$ with $|\omega|=1$ are polar coordinates (i.e.\ $x=r\omega$, $r=|x|$, $\omega=x/|x|$). The function $Y_k$ varies over the space of spherical harmonics of degree $k$, i.e.\ $r^k Y_k(\omega)$ is a $k$-homogeneous harmonic polynomial on $\R^3$. Therefore, $\epsilon_k$ has multiplicity $(2k+1)$.

It is clear that these functions are solutions, and since the $Y_k$ span  $L^2(S^2)$, there can be no further solutions.

Since the problem is scale invariant, one gets the same $\epsilon$ for spheres of arbitrary radius. This implies that
for $a\equiv 1$ we should get $\epsilondot_k = 0$. As a check on our calculations we can see this directly:

From $-\Deltad Y_k = k(k+1) Y_k$ we have for $u_-=u_{-,k}$
 $$\intdO |\nablad u_-|^2 = -\intdO u_-\Deltad u_- = k(k+1) \intdO |Y_k|^2.$$
Also, $\partial_n u_- = \partial_r u_- = k Y_k$ at $r=1$, so the right hand side of \eqref{eq:epsdot} equals for $a\equiv 1$
$$ \left[ k(k+1) + \epsilon_k k^2 \right] \intdO |Y_k|^2 = 0.$$

For arbitrary $a$ we calculate $\epsilondot$ in case $k=1$ and $u_-(x,y,z) = z$, where $\epsilon=2$. This choice of $u_-$ (out of the three-dimensional space, spanned by the functions $x,y,z$, of solutions for $k=1$) corresponds to an exterior electric field pointing in the $z$-direction, in the physical problem leading to \eqref{I1}-\eqref{I3}. We parametrize the sphere as graph $z=\pm \sqrt{1-x^2-y^2}$ \ownremark{(probably not the best choice!)}. Since $u_-$ is 1-homogeneous, we have $\partial_r u_- = u_-=z$ at $r=1$. Then $|\nablad u_-|^2 = |\nabla u_-|^2 - |\partial_r u_-|^2 = 1-z^2$. So the bracket in \eqref{eq:epsdot} equals $ (1-z^2) -2z^2 = 1-3z^2 = 3(x^2+y^2) -2$. The volume element of the sphere is $dx\,dy/\sqrt{1-x^2-y^2}$. To normalize $u$ we divide by  $\|u\|_-^2=\int_\Omega |\nabla u_-|^2 = \int_\Omega 1 = \vol(\Omega) = \frac43\pi$. Finally, $\epsilon+1 = 3$. Therefore, we get
\begin{gather*}
\epsilondot_1 = \frac9{4\pi} \int_{x^2+y^2<1} A(x,y) \frac{3(x^2+y^2)-2}{\sqrt{1-x^2-y^2}}\, dx dy ,\\
 A(x,y) := a(x,y,\sqrt{1-x^2-y^2}) + a(x,y,-\sqrt{1-x^2-y^2}) .
\end{gather*}

\subsection{Example: Half space}\label{subsec:ex plane}
We always assumed $\Omega$ to be compact. However, a similar problem can also be formulated for non-compact $\Omega$. For example, consider the half space
$$ \Omega = \{(x',x_n):\ x'\in \R^{n-1},\ x_n<0\}$$
We then consider the problem \eqref{I1}-\eqref{I3}, but we need to replace the decay condition \eqref{I4}. In view of the regularity discussion it is natural to consider $H^{1/2}$ boundary data, hence we demand a uniform $H^{1/2}$ condition: there is a constant $C$ so that 
\begin{equation} \label{eqn:L2 bounded}
 \int_{\R^{n-1}}|u(x',x_n)|^2\,dx' + 
 \int_{\R^{n-1}}|u(x',x_n)|\, |\nabla_{x'} u(x',x_n)|\,dx'
 \leq C \ \text{ for all } x_n  
\end{equation}
In particular, we can take the Fourier transform in the $x'$-variables,
$\uhat(\xi',x_n) = \int_{\R^{n-1}} e^{-ix'\cdot \xi'} u(x',x_n)\,dx'$, and $\Delta u=0$ yields $(\partial_{x_n}^2 - |\xi'|^2)\uhat = 0$ for $x_n\neq0$, with solutions $\uhat_\pm = a_\pm e^{-|\xi'| x_n} + b_\pm e^{|\xi'| x_n}$ (if $\xi'\neq0$). Then \eqref{eqn:L2 bounded} implies $b_+=a_-=0$, and \eqref{I2} gives $b_-=a_+$, so  
$$ u(x',x_n) = \int_{\R^{n-1}} e^{ix'\cdot\xi'} h(\xi') e^{-|\xi'|\, |x_n|} \,d\xi' $$
for a function $h$ on $\R^{n-1}$ satisfying $\int_{\R^{n-1}} |h(\xi')|^2 (1+|\xi'|)\, d\xi' < \infty$. We have $h=(2\pi)^{-n} \ghat$ for the boundary value $g(x')=u(x',0)$. Also $\partial_n u_\pm (x',0) = \mp\int e^{ix'\cdot\xi'} |\xi'| h(\xi')\, d\xi'$ (distributional Fourier transform), which shows that $\eps=1$ for any $g$.

Therefore, $\eps=1$ is the only plasmonic eigenvalue for the half space.
Incidentally, the explicit formulas show that the Dirichlet-Neumann operator is precisely the Fourier multiplier by $|\xi'|$, i.e. $\calN_-g = (|\xi'|\ghat)\check{ }$.

\ownremark{Open problems:
\begin{enumerate}
\item
Precise asymptotics of $\eps_k$ as $k\to\infty$.
\item
Asymptotics of $\eps_1$ for two domains approaching each other.
\item
What's $\eps_1$ for a periodic chain, or for perturbations of such?
\item
Inverse problem??
\item
Spectral geometry for $\eps_1$.
\end{enumerate}
\begin{conjecture}\label{conj:eps=1}
Low regularity plasmons with $\eps=1$ are smooth at boundary points where the curvature (WHICH CURVATTURE?) of $\dO$ does not vanish. If the curvature does not vanish anywhere then the space of plasmons with $\eps=1$  is finite dimensional.
\end{conjecture}
}

\end{document}